\documentclass[journal,12pt,draftclsnofoot,onecolumn,a4paper]{IEEEtran}
\usepackage[pdftex]{graphicx}
\usepackage[hyphens]{url}
\usepackage{mdwlist}
\usepackage{amssymb}
\usepackage{amsmath}
\usepackage{mathtools}
\usepackage{lscape}
\usepackage{multirow}
\usepackage{graphics}
\usepackage{stfloats}
\usepackage{epstopdf}
\usepackage{amsthm}
\usepackage{algorithm2e}
\usepackage{subcaption}
\usepackage[noadjust,compress]{cite}

\RestyleAlgo{boxruled}

\theoremstyle{plain}
\newtheorem{theorem}{Theorem}
\newtheorem{lemma}{Lemma}

\begin{document}
\title{Physical-layer Network Coding in Two-Way Heterogeneous Cellular Networks with Power Imbalance}

\author{Ajay Thampi, \IEEEmembership{Student Member, IEEE,} Soung Chang Liew, \IEEEmembership{Fellow, IEEE,} Simon Armour, Zhong Fan, Lizhao You, \IEEEmembership{Student Member, IEEE,} Dritan Kaleshi
}

\maketitle

\begin{abstract}
The growing demand for high-speed data, quality of service (QoS) assurance and energy efficiency has triggered the evolution of 4G LTE-A networks to 5G and beyond. Interference is still a major performance bottleneck. This paper studies the application of physical-layer network coding (PNC), a technique that exploits interference, in heterogeneous cellular networks. In particular, we propose a rate-maximising relay selection algorithm for a single cell with multiple relays based on the decode-and-forward strategy. With nodes transmitting at different powers, the proposed algorithm adapts the resource allocation according to the differing link rates and we prove theoretically that the optimisation problem is log-concave. The proposed technique is shown to perform significantly better than the widely studied selection-cooperation technique. We then undertake an experimental study on a software radio platform of the decoding performance of PNC with unbalanced SNRs in the multiple-access transmissions. This problem is inherent in cellular networks and it is shown that with channel coding and decoders based on multiuser detection and successive interference cancellation, the performance is better with power imbalance. This paper paves the way for further research in multi-cell PNC, resource allocation, and the implementation of PNC with higher-order modulations and advanced coding techniques.
\end{abstract}


\IEEEpeerreviewmaketitle

\section{Introduction}
\label{sec:introduction}
As operators evolve their networks toward the 4$^{\text{th}}$ generation (4G) Long Term Evolution-Advanced (LTE-A), the research community has moved on to the study of technologies to be adopted in the 5$^{\text{th}}$ generation (5G). The evolution to 5G is triggered by the forecast explosion in mobile data traffic, which is expected to grow 7-fold from 2013 to 2017 \cite{Cisco2013}; 66\% of that mobile traffic is expected to be video by 2017 with an increasing number of devices requiring high-speed wireless broadband.

In LTE-A systems, attempts to address these requirements are made by cell size reduction and aggressive frequency reuse. As a result, interference between cell sites is identified as the major performance bottleneck \cite{Akyildiz2010,Gesbert2010} and techniques such as coordinated multipoint (CoMP) transmission and reception and heterogeneous networks (HetNet) have been proposed \cite{Hossein2011}. In a CoMP-based HetNet, the base station and users coordinate their transmissions and receptions with the help of many low-powered nodes such as relays, femtocells, picocells and remote radio heads. Such systems are shown to have improved cell coverage and also spectral efficiency \cite{Peng2011,Bhat2012}. Further performance gains can be achieved by employing physical-layer network coding (PNC).

PNC was first proposed in 2006 \cite{Zhang2006a,Popovski2006} as a way to exploit interference inherent in wireless communication systems. Rather than treating interference as a form of corruption, PNC exploits the natural network coding operation that occurs when the desired and interfering electromagnetic waves superimpose with each other. Compared with the traditional non-network-coded scheme (TS), PNC could achieve a 100\% throughput gain \cite{Liew2013}. Since its inception, PNC has gained a wide following in the research community and has recently been considered as a study item in the 3GPP standards \cite{3GPP2009a,3GPP2009,Osseiran2011}.

\subsection{Related Work}
\label{sec:related_work}
To date, most PNC studies have focused on the two-way relay channel (TWRC) model where all the nodes transmit at equal powers \cite{Liew2013}. Two key issues in PNC, symbol asynchrony and channel coding, were addressed in the time domain in \cite{Lu2012} and in the frequency domain in \cite{Lu2013a}. PNC was also successfully implemented on a software radio platform and insights on throughput gains, symbol misalignment, channel coding, effect of carrier frequency offset and real-time issues were gained through these practical prototyping efforts \cite{Lu2013a,Lu2013b,Wu2014}. 

In cellular networks, where there are multiple relays deployed in the cell, an important problem is to select the optimum relay to assist the end-to-end information exchange between the base station and the user. In \cite{Louie2010} and \cite{Song2010}, relay selection was studied in a PNC system where the amplify-and-forward\footnote{The relay amplifies the received superimposed network coded symbol and forwards it to the end nodes.} (AF) strategy was adopted at the relays. Both relay selection algorithms were based on minimising the overall sum bit-error-rate, and the optimisation problem was simplified by assuming equal time allocation for all the links. An AF-based PNC system is however known to be limited by noise, especially at low signal-to-noise ratios (SNRs) \cite{Katti2007}. For instance, when the SNR is between 5-7.5 dB in a symmetric TWRC, the achievable rate of ANC reduces by 5-22\% when compared to TS \cite{Liew2013}. The performance limitation due to noise could be mitigated by adopting the decode-and-forward\footnote{The relay decodes the superimposed network coded symbol rather than the individual symbols transmitted by the end nodes.} (DF) strategy at the relays \cite{Liew2013}. At low SNRs (5-7.5 dB) in a symmetric TWRC, DF-based PNC still performs better than TS, achieving a rate gain of 20-27\%. The DF strategy is therefore considered in this paper. 

Relay selection in a DF-based PNC system was studied in \cite{Ju2010}. The algorithm, called SC-PNC, is based on the widely applied selection-cooperation technique \cite{Bletsas2007,Beres2008} and consists of two steps. In the first step, the end nodes transmit their symbols in the multiple-access phase and all the relays that are not in outage are added to a list for selection. In the second step, the relay in the list that minimises the broadcast-phase outage probability or maximises the minimum mutual information of the two broadcast links is selected. This algorithm assumes equal time allocation for all the links and a closed-form expression for the outage probability is derived. 

The drawback of the approach in \cite{Ju2010} is that the relay selected to maximise the minimum mutual information of the two broadcast links may not be the optimum one for the multiple-access phase. This sub-optimum selection could affect the overall rate of the PNC system. We have also seen that the relay selection algorithms in the literature are simplified by assuming equal time allocation for all the links. The performance of the system could be further improved by allocating more time for the weaker link. In addition, the problem of power imbalance, which is inherent in a cellular network, has not been studied. Since all the nodes transmit at different powers, the decoding performance at the relay in the multiple-access phase could be impacted. All the above gaps are addressed in this paper.

\subsection{Contributions}
\label{sec:contributions}
We consider a PNC system where the nodes transmit at different powers and the time slots allocated for the links are made inversely proportional to their achievable rates. Our objective is to maximise the overall rate of the PNC system and with imbalanced transmitted powers, this necessitates allocating more time for the node with the weaker link. To the best of our knowledge, such a system has not been studied in the literature and we call it PNC-B. We prove that the optimisation problem is log-concave and propose a gradient-ascent based algorithm for relay selection. The performance of PNC-B, in terms of overall rate and densification gain, is shown to be much better than SC-PNC \cite{Ju2010}.

We then study the decoding performance of the relay in the multiple-access phase, given the power imbalance in the system. An experimental study on a software radio platform is conducted. We show that with link-by-link channel coding, the decoding success rate is better when there is an imbalance in power. In addition, we show that power control to balance the SNRs could be detrimental to the performance, especially at low SNRs.

The rest of the paper is structured as follows. Section \ref{sec:system_model} gives an overview of the system model adopted in this paper. Section \ref{sec:tx_strategies_rates} then studies the transmission strategies and their corresponding information-theoretic rates. In Section \ref{sec:relay_selection_pncb}, we look at the relay selection problem for PNC-B. The proof that the optimisation problem is log-concave and the derivation of the algorithm can be found in Section \ref{sec:pncb_algorithm}. The simulations results comparing the performance of the proposed PNC-B algorithm with SC-PNC \cite{Ju2010} can be found in Section \ref{sec:simulation_results}. Section \ref{sec:pnc_decoding_accuracy} then describes the software radio experimental setup and analyses the decoding performance of PNC-B for various SNRs. Finally, Section \ref{sec:conclusions} concludes the paper and suggests avenues for further research.

\section{System Model}
\label{sec:system_model}
The system consists of a cell served by a single base station with multiple users and relays. The traffic between the base station and the users is bidirectional. We assume that every scheduled user will have a unique relay assisting it. Both the linear and planar network models are considered, as shown in Figures \ref{fig:linear_model} and \ref{fig:planar_model} respectively. 

\begin{figure}[h!]
\centering
\begin{subfigure}{.5\textwidth}
  \centering
  \includegraphics[width=0.85\linewidth]{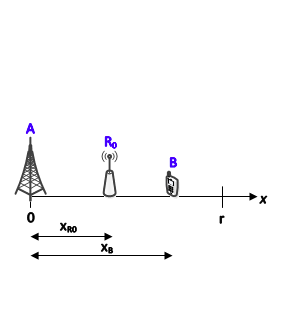}
  \caption{Linear Model}
  \label{fig:linear_model}
\end{subfigure}%
\begin{subfigure}{.5\textwidth}
  \centering
  \includegraphics[width=1.0\linewidth]{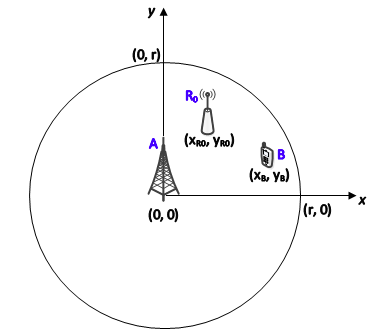}
  \caption{Planar Model}
  \label{fig:planar_model}
\end{subfigure}
\caption{Network Models}
\label{fig:test}
\end{figure}

The figures show a single scheduled user and a unique relay assisting it. In both models, the base station is represented as node $A$, the user as node $B$ and the optimum relay assisting them as node $R_0$. Each node is equipped with a single omnidirectional antenna. The cell radius is denoted by $r$ and the base station is placed at the origin in both models. In the linear model, the relay and user are at distances $x_{R_0}$ and $x_B$ respectively from the base station. In the planar model, the locations of the relay and user are Cartesian coordinates, $(x_{R_0},y_{R_0})$ and $(x_B,y_B)$ respectively.

In general, the received power at node $y$ when node $x$ transmits at power $P_x^{(t)}$ is given by
\begin{equation}
P_{xy}^{(r)} = \bar{P}_x |h_{xy}|^2 d_{xy}^{-n}
\label{eqn:p_rx_xy}
\end{equation}
where $n$ is the path loss exponent, and $|h_{xy}|$ and $d_{xy}$ are the normalised gain of the channel and the distance between nodes $x$ and $y$, respectively. In (\ref{eqn:p_rx_xy}), $\bar{P}_x$ is the received power from node $x$ accounting for the free space path loss, given by
\begin{equation}
\bar{P}_x = \left( \frac{c}{4 \pi f_c} \right)^2 d_0^{n-2} P_x^{(t)}
\label{eqn:p_bar_x}
\end{equation}
where $c$ is the speed of light in vacuum, $f_c$ is the carrier frequency, $d_0$ is the reference distance and $P_x^{(t)}$ is the transmitted power of node $x$. 

In cellular networks, it is fair to assume the following constraint on the transmitted powers.
\begin{equation}
P_A^{(t)} > P_{R_0}^{(t)} > P_B^{(t)}
\label{eqn:txd_power_constraint}
\end{equation}

This form of power imbalance is considered in this paper. Without loss of generality, time-division duplexing is assumed and all nodes in the network respect the half-duplex constraint since full-duplex wireless is presently very challenging to implement.

\section{Transmission Strategies and Rates}
\label{sec:tx_strategies_rates}
The PNC scheme is shown in Figure \ref{fig:pnc_strategy}. In the first time slot, called the multiple-access phase, the base station and user (nodes A and B respectively) transmit simultaneously. The relay tries to deduce a network coded message from the superimposed signals of A and B in the multiple-access phase. This process is called PNC mapping and is described in great detail in \cite{Liew2013}. In the second time slot, called the broadcast phase, the relay broadcasts the deduced network-coded message to the base station and the user. Using the self-information, each end node can extract the signal transmitted by the other.

\begin{figure}[h!]\centering\includegraphics[width=0.4\textwidth]{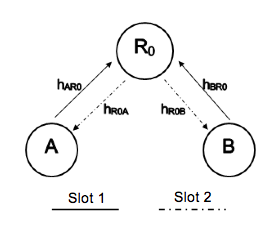}\caption{Physical-Layer Network Coding}\label{fig:pnc_strategy}\end{figure}

The rate of the multiple-access phase is upper-bounded by (\ref{eqn:r_ma}). If link-by-link channel coding is done in the PNC system, where the relay performs channel decoding and re-encoding in addition to PNC mapping \cite{Liew2013}, then it is shown in \cite{Nam2010} that the upper bound is approached within 1/2 bit using nested lattice codes. The rate of the multiple-access phase using lattice codes is given by (\ref{eqn:r_ma_lc}).

\begin{equation}
\begin{aligned}
R_{MA} &= \min \left\{ \log_2 \left( 1 + \frac{P_{AR_0}^{(r)}}{N_0 W} \right) ,\right. \\ &\ \ \ \ \ \ \ \ \ \ \ \left. \log_2 \left( 1 + \frac{P_{BR_0}^{(r)}}{N_0 W} \right)  \right\} \text{bps/Hz}
\end{aligned}
\label{eqn:r_ma}
\end{equation}

\begin{equation}
\begin{aligned}
R_{MA}^{LC} &= \min \left\{ \log_2 \left( \frac{P_{AR_0}^{(r)}}{P_{AR_0}^{(r)} + P_{BR_0}^{(r)}} + \frac{P_{AR_0}^{(r)}}{N_0 W} \right) , \right. \\ &\ \ \ \ \ \ \ \ \ \ \left. \log_2 \left( \frac{P_{BR_0}^{(r)}}{P_{AR_0}^{(r)} + P_{BR_0}^{(r)}} + \frac{P_{BR_0}^{(r)}}{N_0 W} \right)  \right\} \text{bps/Hz}
\end{aligned}
\label{eqn:r_ma_lc}
\end{equation}

For the broadcast phase, the rate is given by (\ref{eqn:r_bc}).
\begin{equation}
\begin{aligned}
R_{BC} &= \min \left\{ \log_2 \left( 1 + \frac{P_{R_0A}^{(r)}}{N_0 W} \right) , \right. \\ &\ \ \ \ \ \ \ \ \ \ \ \left. \log_2 \left( 1 + \frac{P_{R_0B}^{(r)}}{N_0 W} \right)  \right\} \text{bps/Hz}
\end{aligned}
\label{eqn:r_bc}
\end{equation}

The overall achievable rate of the PNC system with equal time-slot allocation, as assumed in the literature \cite{Louie2010,Song2010,Ju2010}, is given by (\ref{eqn:r_pnc}). 
\begin{equation}
R_{PNC} = \frac{1}{2} \min \left\{ R_{MA}^{LC} , R_{BC} \right\} \text{bps/Hz}
\label{eqn:r_pnc}
\end{equation}
For the PNC system considered in this paper with rate-maximizing unbalanced time allocation, the overall achievable rate is given by (\ref{eqn:r_pncb2}).
\begin{equation}
R_{PNC-B} = \min \left\{ \rho_{MA} R_{MA}^{LC} , \rho_{BC} R_{BC} \right\} \text{bps/Hz}
\label{eqn:r_pncb2}
\end{equation}
where $\rho_{MA} = \frac{R_{BC}}{R_{MA}^{LC} + R_{BC}}$ and $\rho_{BC} = \frac{R_{MA}^{LC}}{R_{MA}^{LC} + R_{BC}}$ are the fractions of time allocated for the multiple-access and broadcast phases respectively. Thus,
\begin{equation}
R_{PNC-B} = \frac{R_{MA}^{LC} R_{BC}}{R_{MA}^{LC} + R_{BC}} \text{bps/Hz}
\label{eqn:r_pncb}
\end{equation}

\section{Relay Selection for PNC-B}
\label{sec:relay_selection_pncb}
\subsection{The Algorithm}
\label{sec:pncb_algorithm}
This sub-section addresses the problem of relay selection for the PNC-B transmission strategy. To the best of our knowledge, this has not been studied in the literature. The SC-PNC approach in \cite{Ju2010} cannot be easily extended to PNC-B as the derivation of the outage probability becomes mathematically intractable.

The optimum relay is the one that maximises the overall rate of the PNC-B system given by (\ref{eqn:r_pncb_2}). 
\begin{equation}
\tilde{R}_{PNC-B} = \frac{R_{MA} R_{BC}}{R_{MA} + R_{BC}} \text{bps/Hz}
\label{eqn:r_pncb_2}
\end{equation}
In (\ref{eqn:r_pncb_2}), the upper bound rate for the multiple-access phase, given by (\ref{eqn:r_ma}), is used for analytical simplicity. If the rate using nested lattice codes, given by (\ref{eqn:r_pncb}) is used, the optimisation problem becomes mathematically intractable. It will however be shown in Section \ref{sec:simulation_results} that the optimum relay derived using the upper-bound approximation in turn maximises the achievable rate using nested lattice codes.

$R_{MA}$ and $R_{BC}$, given by (\ref{eqn:r_ma}) and (\ref{eqn:r_bc}) respectively, can be rewritten as $R_{MA} = \min \left\{ R_{AR_0} , R_{BR_0} \right\}$ and $R_{BC} = \min \left\{ R_{R_0A} , R_{R_0B} \right\}$. In order to keep the equations simple, we first consider the linear model.

\begin{lemma}
The overall rate of the PNC-B system, dependent on the user and relay locations, consists of four cases given by,
\begin{equation}
\begin{aligned}
\tilde{R}_{PNC-B} = \left\{ 
  \begin{array}{l l}
    \frac{R_{AR_0} R_{R_0A}}{R_{AR_0} + R_{R_0A}} & \quad \text{if } x_{R_0} \in \left(\frac{x_B}{D},x_B\right] \\
    \frac{R_{AR_0} R_{R_0B}}{R_{AR_0} + R_{R_0B}} & \quad \text{if } x_{R_0} \in \left(\frac{x_B}{D},\frac{x_B}{2}\right] \\
    \frac{R_{BR_0} R_{R_0A}}{R_{BR_0} + R_{R_0A}} & \quad \text{if } x_{R_0} \in \left(\frac{x_B}{2},\frac{x_B}{D}\right] \\
    \frac{R_{BR_0} R_{R_0B}}{R_{BR_0} + R_{R_0B}} & \quad \text{if } x_{R_0} \in \left[ d_0,\frac{x_B}{2} \right] 
  \end{array} \right.
\end{aligned}
\label{eqn:r_pncb_cases}
\end{equation}
where $D = 1 + \left( P_B^{(t)} / P_A^{(t)}\right)^{1 / n}$.
\label{lem:pnc_b_ranges}
\end{lemma}
\begin{proof}
From (\ref{eqn:r_ma}), $R_{MA} = R_{AR_0}$ when $R_{AR_0} < R_{BR_0}$ and $R_{MA} = R_{BR_0}$ otherwise. When $R_{AR_0} < R_{BR_0}$,
\begin{equation}
\begin{aligned}
&\log_2 \left( 1 + \frac{P_{AR_0}^{(r)}}{N_0 W} \right) < \log_2 \left( 1 + \frac{P_{R_0B}^{(r)}}{N_0 W} \right) \\
&\log_2 \left( 1 + \frac{\bar{P}_A |h_{AR_0}|^2 x_{R_0}^{-n}}{N_0 W} \right) < \\&\log_2 \left( 1 + \frac{\bar{P}_B |h_{BR_0}|^2 \left|x_B - x_{R_0}\right|^{-n}}{N_0 W} \right)
\end{aligned}
\label{eqn:r_ar0_less_br0}
\end{equation}
Ignoring the effects of fading and considering that the optimum relay has to lie between the base station and the user, (\ref{eqn:r_ar0_less_br0}) reduces to
\begin{equation}
P_A^{(t)} x_{R_0}^{-n} < P_B^{(t)} (x_B - x_{R_0})^{-n}
\label{eqn:r_ar0_less_br0_2}
\end{equation}

The effects of fading are ignored to decouple the problem of relay selection from resource allocation. In an OFDM system, the effects of wideband (frequency-selective) fading are mitigated by dividing the signal into many narrowband subcarriers. Maximising performance in such a flat fading environment is then a resource allocation issue. In \cite{Xu2013}, the problem of subcarrier allocation in an OFDMA-based heterogeneous system that employed straightforward network coding was studied. For a PNC system, this problem would alter slightly since the same subcarriers would need to be used by the base station and the user in the multiple-access phase. This is however beyond the scope of this paper and will be addressed in the future.

Solving (\ref{eqn:r_ar0_less_br0_2}) for $x_{R_0}$, we can obtain the range of values for which $R_{MA} = R_{AR_0}$ which are, $x_{R_0} \in \left(\frac{x_B}{D},x_B\right)$ where $D = 1 + \left( P_B^{(t)} / P_A^{(t)}\right)^{1 / n}$. On the other hand, $R_{MA} = R_{BR_0}$ when $x_{R_0} \in \left[d_0,\frac{x_B}{D}\right]$.

Similarly, from (\ref{eqn:r_bc}), $R_{BC} = R_{R_0A}$ when $R_{R_0A} < R_{R_0B}$ and $R_{BC} = R_{R_0B}$ otherwise. Thus from $R_{R_0A} < R_{R_0B}$, we get
\begin{equation}
P_{R_0}^{(t)} x_{R_0}^{-n} < P_{R_0}^{(t)} (x_B - x_{R_0})^{-n}
\label{eqn:r_r0a_less_r0b_2}
\end{equation}
Equation (\ref{eqn:r_r0a_less_r0b_2}) is obtained similar to (\ref{eqn:r_ar0_less_br0_2}). Solving (\ref{eqn:r_r0a_less_r0b_2}) for $x_{R_0}$, we get $R_{BC} = R_{R_0A}$ when $x_{R_0} \in \left(\frac{x_B}{2},x_B\right)$ and $R_{BC} = R_{R_0B}$ when $x_{R_0} \in \left[d_0,\frac{x_B}{2}\right]$. Combining these results for $R_{MA}$ and $R_{BC}$, we can obtain (\ref{eqn:r_pncb_cases}).
\end{proof}

\begin{lemma}
The search for the optimum relay can be restricted to the range $\left(\frac{x_B}{2},\frac{x_B}{D}\right]$.
\label{lem:optimum_relay_range}
\end{lemma}
\begin{proof}
To prove this lemma, the validity of the four cases in Lemma \ref{lem:pnc_b_ranges} have to be analysed. Given the power constraint (\ref{eqn:txd_power_constraint}) where $P_A^{(t)} > P_B^{(t)}$, the constant $D$ is less than 2. Hence, $\frac{x_B}{D} > \frac{x_B}{2}$ which makes Case 2, where $x_{R_0} \in \left(\frac{x_B}{D},\frac{x_B}{2}\right]$, invalid. For Case 1 where $x_{R_0} \in \left(\frac{x_B}{D},x_B\right]$, the optimisation problem will be skewed towards the user since a relay closer to the user will be chosen. Similarly, Case 4 will be skewed towards the base station since a relay in the range $\left[ d_0,\frac{x_B}{2} \right]$ will be chosen. Thus, the search space for the optimum relay must be the one in Case 3 where $x_{R_0} \in \left(\frac{x_B}{2},\frac{x_B}{D}\right]$.
\end{proof}

Using Lemmas \ref{lem:pnc_b_ranges} and \ref{lem:optimum_relay_range}, (\ref{eqn:r_pncb_2}) simplifies to
\begin{equation}
\tilde{R}_{PNC-B} = \frac{R_{BR_0} R_{R_0A}}{R_{BR_0} + R_{R_0A}}\ ;\ x_{R_0} \in \left(\frac{x_B}{2},\frac{x_B}{D}\right]
\label{eqn:r_pncb_3}
\end{equation}
The objective function for PNC-B is then
\begin{equation}
f(x_{R_0}) = \left(\frac{1}{\log_e 2}\right) . \frac{g(x_{R_0}) h(x_{R_0})}{g(x_{R_0}) + h(x_{R_0})}
\label{eqn:obj_function_pncb}
\end{equation}
In (\ref{eqn:obj_function_pncb}), 
\begin{equation}
g(x_{R_0}) = \log_e \left(1 + \Gamma_B (x_B - x_{R_0})^{-n}\right)
\label{eqn:g_2}
\end{equation}
\begin{equation}
h(x_{R_0}) = \log_e \left(1 + \Gamma_{R_0} x_{R_0}^{-n} \right)
\label{eqn:h_2}
\end{equation}
where $\Gamma_B = \frac{\bar{P}_B |h_{BR_0}|^2}{N_0 W}$ and $\Gamma_{R_0} = \frac{\bar{P}_{R_0}|h_{R_0A}|^2}{N_0 W}$. The relay selection problem for the linear model can be formulated as
\begin{equation}
\begin{aligned}
& \underset{x_{R_0}}{\text{maximise}}
& & f(x_{R_0}) \\
& \text{subject to}
& & \frac{x_B}{2} < x_{R_0} \le \frac{x_B}{D} \\
&&& \left|x_B - x_{R_0}\right| \ge d_0 \\
\end{aligned}
\label{eqn:pncb_opt_problem_linear}
\end{equation}
The following lemmas and theorem will help design the algorithm to solve (\ref{eqn:pncb_opt_problem_linear}).
\begin{lemma}
If a function $f$ on $\mathbb{R}$ is twice differentiable, then it is log-concave if and only if $\textbf{dom} f$ is a convex set and $f^{\prime \prime}(x) f(x) \le f^{\prime}(x)^2, \forall x \in \textbf{dom} f$ \cite{Boyd2004}.
\label{lem:log_convex_concave}
\end{lemma}
\begin{lemma}
Log-convexity and log-concavity are closed under multiplication and positive scaling \cite{Boyd2004}.
\label{lem:log_convex_concave_preserve}
\end{lemma}
\begin{theorem}
The objective function for PNC-B is log-concave for $x_{R_0} \in \left(\frac{x_B}{2},\frac{x_B}{D}\right]$.
\label{thm:f_pncb_logconcave}
\end{theorem}
\begin{proof}
Taking the logarithm of (\ref{eqn:obj_function_pncb}),
\begin{equation}
\begin{aligned}
F &= \log_e g + \log_e h - \log_e \left(g + h\right) - \log_e \left(\log_e 2\right)
\end{aligned}
\end{equation}

To prove log-concavity, we need to obtain the first and second order derivativse of $R_{BR_0}$ and $R_{R_0A}$, which are functions $g(x_{R_0})$ and $h(x_{R_0})$ respectively. The first-order derivative of $g(x_{R_0})$ can be obtained to be
\begin{equation}
g^{\prime}(x_{R_0}) = \frac{n \Gamma_B \left(x_B - x_{R_0}\right)^{-n-1}}{1 + \Gamma_B \left(x_B - x_{R_0}\right)^{-n}}
\label{eqn:gprime_1}
\end{equation}
An assumption that $\Gamma_B \left(x_B - x_{R_0}\right)^{-n} \gg 1$ is made, which is valid for the setup considered. Moreover, this assumption only serves to keep the equations simpler and does not affect the proof for log-concavity. Equation (\ref{eqn:gprime_1}) becomes (\ref{eqn:gprime_2}) and the second-order derivative is (\ref{eqn:gpprime_1}). 
\begin{equation}
g^{\prime}(x_{R_0}) = \frac{n}{x_B - x_{R_0}}
\label{eqn:gprime_2}
\end{equation}
\begin{equation}
g^{\prime \prime}(x_{R_0}) = - \frac{n}{\left(x_B - x_{R_0}\right)^2}
\label{eqn:gpprime_1}
\end{equation}
Now for $h(x_{R_0})$, the first-order derivative is
\begin{equation}
h^{\prime}(x_{R_0}) = - \frac{n \Gamma_{R_0} x_{R_0}^{-n-1}}{1 + \Gamma_{R_0} x_{R_0}^{-n}}
\label{eqn:hprime_1}
\end{equation}
A similar simplification is made as earlier by assuming that $\Gamma_{R_0} x_{R_0}^{-n} \gg 1$, which is valid for the setup considered. Equation (\ref{eqn:hprime_1}) becomes (\ref{eqn:hprime_2}) and the second-order derivative is (\ref{eqn:hpprime_1}).
\begin{equation}
h^{\prime}(x_{R_0}) = - \frac{n}{x_{R_0}}
\label{eqn:hprime_2}
\end{equation}
\begin{equation}
h^{\prime \prime}(x_{R_0}) = \frac{n}{x_{R_0}^2}
\label{eqn:hpprime_1}
\end{equation}

Using (\ref{eqn:g_2}), (\ref{eqn:gprime_2}) and (\ref{eqn:gpprime_1}), we can obtain
\begin{equation}
\begin{aligned}
g(x_{R_0}) g^{\prime \prime}(x_{R_0}) = - \frac{n \log_e \left(1 + \Gamma_B (x_B - x_{R_0})^{-n}\right)}{\left(x_B - x_{R_0}\right)^2}
\end{aligned}
\end{equation}
and
\begin{equation}
g^{\prime}(x_{R_0})^2 = \left(\frac{n}{x_B - x_{R_0}}\right)^2
\end{equation}
Since $n$, $\Gamma_B$ and $(x_B - x_{R_0})$ are positive, $g g^{\prime \prime} < 0$ and $g^{\prime} > 0$. Thus by Lemma \ref{lem:log_convex_concave}, $g$ is log-concave.

Similarly, using (\ref{eqn:h_2}), (\ref{eqn:hprime_2}) and (\ref{eqn:hpprime_1}), we can obtain
\begin{equation}
h(x_{R_0}) h^{\prime \prime}(x_{R_0}) = \frac{n \log_e \left(1 + \Gamma_{R_0} x_{R_0}^{-n} \right)}{x_{R_0}^2}
\end{equation}
and
\begin{equation}
h^{\prime}(x_{R_0})^2 = \left(\frac{n}{x_{R_0}}\right)^2
\end{equation}
The condition for log-concavity, $h h^{\prime \prime} \le \left(h^{\prime}\right)^2$, is satisfied if and only if $x_{R_0} \ge \left(\frac{\Gamma_{R_0}}{e^n - 1}\right)^{\frac{1}{n}}$. For the setup considered in this paper, $\left(\frac{\Gamma_{R_0}}{e^n - 1}\right)^{\frac{1}{n}} < d_0$ and since this is outside the domain of $f$, $h$ is also log-concave (by Lemma \ref{lem:log_convex_concave}).

Now let $j = g + h$.  In general, the sum of log-concave functions is not log-concave \cite{Boyd2004}. So we look at the first and second-order derivates of $j$ given by
\begin{equation}
\begin{aligned}
j^{\prime}(x_{R_0}) &= \frac{n \left(2x_{R_0} - x_B\right)}{x_{R_0} \left(x_B - x_{R_0}\right)}
\end{aligned}
\end{equation}
and
\begin{equation}
\begin{aligned}
j^{\prime \prime}(x_{R_0}) &= - \frac{n x_B \left(2x_{R_0} - x_B\right)}{x_{R_0}^2 \left(x_B - x_{R_0}\right)^2}
\end{aligned}
\end{equation}
Since $\textbf{dom}\ f = \left(\frac{x_B}{2},\frac{x_B}{D}\right]$, $2x_{R_0} - x_B > 0$. This means that $j^{\prime} > 0$ and $j^{\prime \prime} < 0$. Thus, $j j^{\prime \prime} \le (j^{\prime})^2$ and hence $j = g + h$ is log-concave (by Lemma \ref{lem:log_convex_concave}).

Since $g$, $h$ and $g + h$ are log-concave, then by Lemma \ref{lem:log_convex_concave_preserve}, $f$ is also log-concave.
\end{proof}
The following observations can be made from Theorem \ref{thm:f_pncb_logconcave}:
\begin{enumerate}
\item An optimisation algorithm such as gradient ascent will be able to find the global optimum relay location, $x_{R_0}^*$.
\item If there is no relay at $x_{R_0}^*$, then the next best option is to choose the one closest to the global optimum solution.
\item The boundary for placing relays for PNC-B is $\frac{r}{D}$, which will be useful for network planning.
\end{enumerate}
The gradient of $F(x_{R_0})$ can be obtained to be
\begin{equation}
\begin{aligned}
F^{\prime} &= \frac{g^{\prime}}{g} + \frac{h^{\prime}}{h} - \frac{g^{\prime} + h^{\prime}}{g + h}
\end{aligned}
\label{eqn:F_prime}
\end{equation}
In order to compute $F^{\prime}$, $\Gamma_{R_0}$ and $\Gamma_B$ have to be obtained which would require channel estimates at the base station and relay respectively. Algorithm \ref{alg:relay_selection_pncb_linear}, which is based on gradient ascent, describes the relay selection process for PNC-B. The parameter $\alpha$ is the step size and the criteria for convergence are:
\begin{enumerate}
\item $F$ at iteration $i+1$ is less than that of iteration $i$, or
\item $x_{R_0}^* > \frac{\hat{x}_B}{D}$
\end{enumerate}
\begin{algorithm}
 Initialise empty list of optimum relays\;
 \For{each user in list}{
  Estimate the user location $\hat{x}_B$ using received SNR $\frac{P_{AB}^{(r)}}{N_0W}$\;
  Initialise $x_{R_0}^* = \frac{\hat{x}_B}{2}$\; 
  \Repeat{
      convergence
    }{$x_{R_0}^* := x_{R_0}^* + \alpha  F^{\prime}$ \tcc{$F^{\prime}$ given by (\ref{eqn:F_prime})}}
  Add relay closest to $x_{R_0}^*$ to the optimum relays list\;
 }
\caption{Relay Selection (Linear Model)}
\label{alg:relay_selection_pncb_linear}
\end{algorithm}
The algorithm requires as input the transmitted powers of each of the nodes and the locations of the deployed relays, which are known a priori. Besides these two inputs, the algorithm also requires the received SNRs from each of the users in the cell in order to estimate its distance from the base station. In the algorithm, the received SNR from the user will be an estimate based on the reference or training symbol transmitted by the base station. This is typically relayed to the base station through the control channel as measurement reports. In LTE, for instance, the measurement report contains the reference signal received power (RSRP) and the reference signal received quality (RSRQ) \cite{3GPP2010a}. In order to obtain the estimate of the user location ($\hat{x}_B$), the operator could employ the Minimisation of Drive Tests (MDT) reports specified in the 3GPP LTE standards \cite{3GPP2012}. These reports contain RSRP, RSRQ and detailed location information in the form of GPS coordinates. This information can be used to train a machine learning algorithm that estimates $\hat{x}_B$. This is however beyond the scope of this paper and will be addressed in the future.

The extension of Algorithm \ref{alg:relay_selection_pncb_linear} to the planar model is straightforward. The objective function $f$ will be dependent on the coordinates $(x_{R_0},y_{R_0})$ and the gradients $\frac{\partial F}{\partial x_{R_0}}$ and $\frac{\partial F}{\partial y_{R_0}}$ have to be computed. Note that in each iteration, $x_{R_0}$ and $y_{R_0}$ have to be updated simultaneously.

\subsection{Simulation Results}
\label{sec:simulation_results}
The simulation setup is summarised in Table \ref{table:sim_setup}. Link-by-link channel coding is done in the PNC system and the achievable rate is computed assuming the use of nested lattice codes \cite{Nam2010} in the system. This rate is averaged over 1000 different network realisations.

\begin{table}[!h]
\caption{Simulation Setup}
\centering
\begin{tabular}{| l | c |}
\hline
Base Station Transmitted Power, $P_A^{(t)}$ & 46 dBm \\ \hline
Relay Transmitted Power, $P_{R_0}^{(t)}$ & 30 dBm \\ \hline
User Transmitted Power, $P_B^{(t)}$ & 23 dBm \\ \hline
Path Loss Exponent, $n$ & 3.7 \\ \hline
Cell Radius, $r$ & 1 km \\ \hline
Reference Distance, $d_0$ & 10 m \\ \hline
Carrier Frequency, $f_c$ & 1.9 GHz \\ \hline
Fading Model & Rayleigh \\ \hline
Step Size, $\alpha$ & 0.01 \\ \hline
\end{tabular}
\label{table:sim_setup}
\end{table}

Figure \ref{fig:relay_selection} shows, as an illustrative example, the achievable rates for different relays for the case of a user located at (850,750) metres, i.e. near the cell edge. Relays, represented as red pluses, are deployed with a separation distance of 200 metres. The achievable rates (in bps/Hz) for the overall system using each relay is shown against the corresponding plus symbol. It can be seen that Algorithm \ref{alg:relay_selection_pncb_linear} for the planar model, which is derived using the upper-bound approximation, chooses the optimum relay that maximises the overall achievable rate of the system.

\begin{figure}[h!]\centering\includegraphics[width=0.8\textwidth]{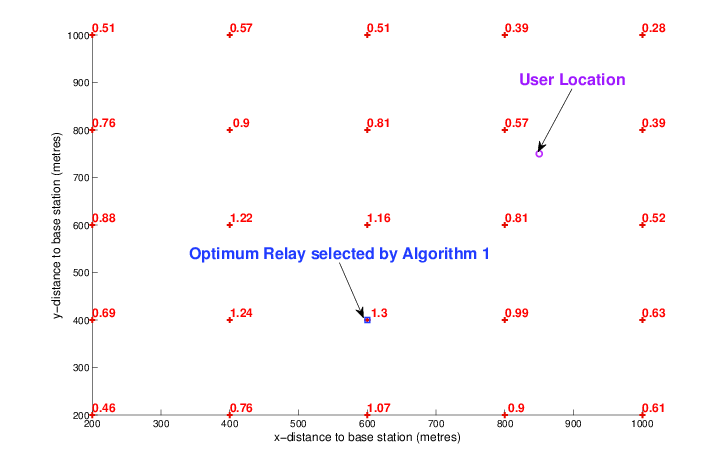}\caption{PNC-B Relay Selection in the Planar Model}\label{fig:relay_selection}\end{figure}

With the relay selection algorithm in place, we compare the rates of the proposed PNC-B scheme with that of SC-PNC \cite{Ju2010} in Figure \ref{fig:rate_comparison}. Two different network deployments are considered, a dense deployment where relays are placed every 10 metres in the cell (Figure \ref{fig:rate_pnc_dt_10}) and a sparse deployment where the relay separation is 400 metres (Figure \ref{fig:rate_pnc_dt_400}). It can be seen that PNC-B outperforms SC-PNC for all user locations. In addition, the gain of PNC-B over SC-PNC is more significant for a sparse deployment. Intuitively, this is down to the unequal time-slot allocation in PNC-B. For a dense deployment, the difference in SNRs between the two multiple-access links will be smaller than that of a sparse deployment.

\begin{figure*}[t!]
\centering
\begin{subfigure}{.5\textwidth}
  \centering
  \includegraphics[width=1.0\linewidth]{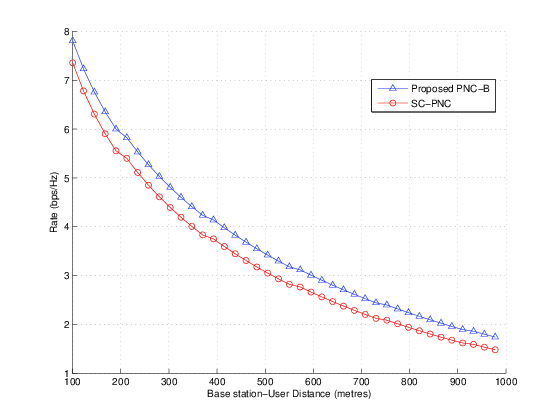}
  \caption{Relay Separation = 10m}
  \label{fig:rate_pnc_dt_10}
\end{subfigure}%
\begin{subfigure}{.5\textwidth}
  \centering
  \includegraphics[width=1.0\linewidth]{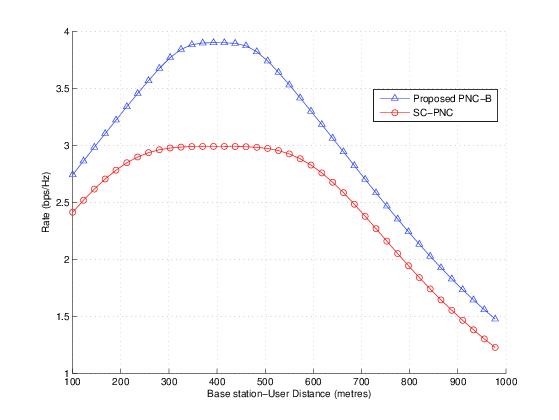}
  \caption{Relay Separation = 400m}
  \label{fig:rate_pnc_dt_400}
\end{subfigure}
\caption{Rate Performance Comparison}
\label{fig:rate_comparison}
\end{figure*}

We will now look at another performance metric, called network densification gain, as defined in \cite{Andrews2014}. The densification gain, $\rho$ in (\ref{eqn:densification_gain}), measures the effective increase in the aggregate data rate relative to the increase in base station or relay density. In (\ref{eqn:densification_gain}), if the number of relays/$\text{km}^2$ is doubled, then the network density factor = 2. 

\begin{equation}
\rho = \frac{\text{Aggregate Rate Gain}}{\text{Network Density Factor}}
\label{eqn:densification_gain}
\end{equation}

In Figure \ref{fig:densification_gains}, the densification gain and rate gain of the proposed PNC-B scheme are compared with that of SC-PNC for various network densities. 100 users were uniformly distributed in the cell and the reference network density was 10 relays/$\text{km}^2$. The y-axis on the left, in blue, represents the densification gain and the y-axis on the right, in red, represents the aggregate rate gain. It  can be observed that if the network density is doubled, PNC-B outperforms SC-PNC by about 22\%. It can also be observed that as the network density increases, we get diminishing returns in terms of the aggregate rate gain. In addition, the rate gain for SC-PNC approaches that of the proposed PNC-B scheme only at very high relay densities. For instance, when the network density is increased 10-fold, the rate is doubled for both PNC-B and SC-PNC, when compared to the reference density of 10 relays/$\text{km}^2$.

\begin{figure}[h!]\centering\includegraphics[width=0.8\textwidth]{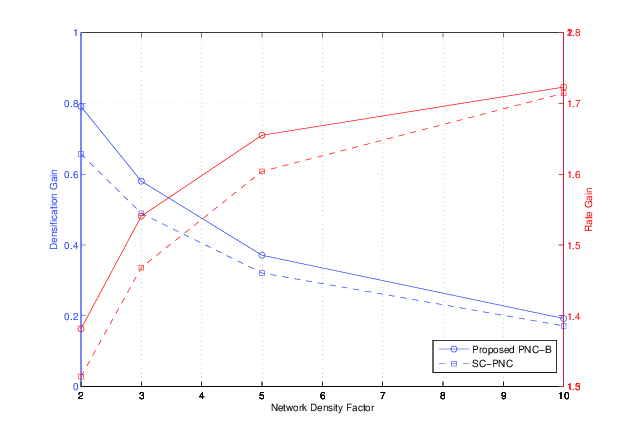}\caption{Densification and Rate Gain Comparison}\label{fig:densification_gains}\end{figure}

\section{PNC Decoding Performance with Power Imbalance}
\label{sec:pnc_decoding_accuracy}
In the previous section, a relay selection algorithm for PNC-B that maximised the overall rate achievable assuming the use of nested lattice codes was proposed. A constraint on the transmitted powers was imposed where $P_A^{(t)} > P_{R_0}^{(t)} > P_B^{(t)}$, which is typical in a cellular network environment. This constraint may however lead to an imbalance in received SNRs at the relay in the multiple-access phase. This section studies the impact of this power imbalance on the decoding performance.

The decoding performance is defined as the rate at which the superimposed signal is successfully decoded by the relay in the multiple-access phase. We perform an experimental study of the decoding performance on the universal soft radio peripheral (USRP) platform, where the conditions of a cellular network are emulated. This study serves to augment the theoretical work done in the previous section and to get us closer to implementing PNC in a practical cellular network.

\subsection{Cumulative Distribution Function of the Received SNRs}
\label{sec:cdf_rxd_snrs}
Before describing the USRP experimental setup, a simulation-based study is done to obtain the cumulative distribution function (CDF) of the received SNRs at the relay for each of the links. The CDF will help us understand the likelihood of power imbalance at the relay in the multiple-access phase. It will also serve to guide the selection of the appropriate SNRs for which experimental results will be collected from the USRP setup. The received SNR of the base station-relay link is denoted by $\Gamma_{AR_0}$ and the received SNR of the user-relay link is denoted by $\Gamma_{BR_0}$. For the setup described in Section \ref{sec:system_model}, 100 users are uniformly distributed in the cell and the optimum relay for PNC-B is selected for each of the users. Two different network deployments are considered: one in which the relays are densely deployed with a separation distance $s$ = 100 m and the other in which $s$ = 600 m. Rayleigh fading is considered and the received SNRs for 1000 different network realisations are obtained. 
\begin{figure*}[t!]
\centering
\begin{subfigure}{.5\textwidth}
  \centering
  \includegraphics[width=1.0\linewidth]{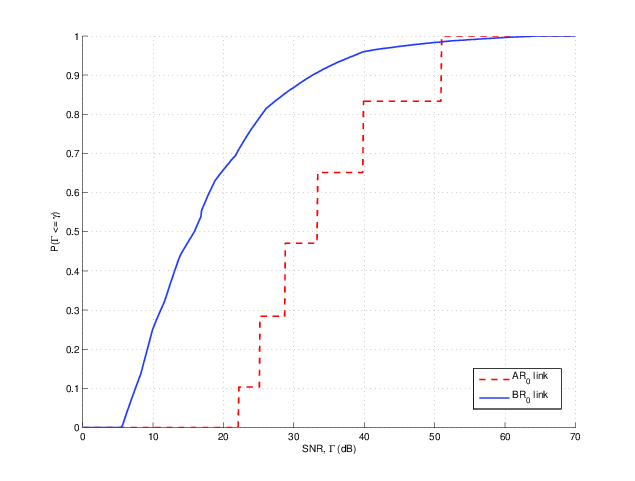}
  \caption{Relay separation, $s$ = 100m}
  \label{fig:cdf_snr_100}
\end{subfigure}%
\begin{subfigure}{.5\textwidth}
  \centering
  \includegraphics[width=1.0\linewidth]{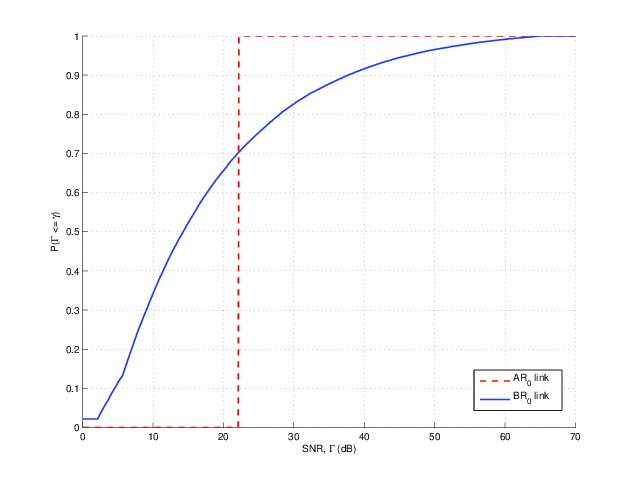}
  \caption{Relay separation, $s$ = 600m}
  \label{fig:cdf_snr_600}
\end{subfigure}
\caption{CDF of SNRs $\Gamma_{AR_0}$ and $\Gamma_{BR_0}$ for PNC-B}
\label{fig:test}
\end{figure*}

Figures \ref{fig:cdf_snr_100} and \ref{fig:cdf_snr_600} show the CDF of $\Gamma_{AR_0}$ and $\Gamma_{BR_0}$ for $s$ = 100 m and $s$ = 600 m respectively. The CDF of $\Gamma_{BR_0}$ is a smooth curve since many users are considered and their locations are randomly distributed in the cell. The CDF of $\Gamma_{AR_0}$ has steps since the location of the base station is fixed and for a given user, there are only a certain number of relays to choose from in the cell. Since there are more choices for relay selection in the dense deployment ($s$ = 100m), $\Gamma_{AR_0}$ has more steps. We define low, medium and high SNRs for a link $xy$ to be the following:
\begin{enumerate}
\item Low SNR: $\Gamma_{xy} \le $ 7.5 dB
\item Medium SNR: 7.5 dB $< \Gamma_{xy} \le$ 10 dB
\item High SNR: 10 dB $< \Gamma_{xy} \le$ 30 dB
\end{enumerate}

\begin{table}[!h]
\caption{Probability of Low, Medium and High SNRs}
\centering
\begin{tabular}{| c || c | c |}
\hline
\multirow{2}{*}{\textbf{CDF}} & \multicolumn{2}{ c |}{\textbf{Relay Separation}} \\ \cline{2-3} 
& \textbf{100 m} & \textbf{600 m} \\ \hline \hline
P($\Gamma_{\text{AR}_0} <= 7.5\ \text{dB}$) & 0 & 0 \\ \hline
P($\Gamma_{\text{BR}_0} <= 7.5\ \text{dB}$) & 0.1 & 0.23 \\ \hline \hline
P($7.5\ \text{dB} < \Gamma_{\text{AR}_0} <= 10\ \text{dB}$) & 0 & 0 \\ \hline
P($7.5\ \text{dB} < \Gamma_{\text{BR}_0} <= 10\ \text{dB}$) & 0.15 & 0.11 \\ \hline \hline
P($10\ \text{dB} < \Gamma_{\text{AR}_0} <= 30\ \text{dB}$) & 0.47 & 1 \\ \hline
P($10\ \text{dB} < \Gamma_{\text{BR}_0} <= 30\ \text{dB}$) & 0.62 & 0.48 \\ \hline
\end{tabular}
\label{table:cdf_snr}
\end{table}
The probabilities for the low, medium and high SNRs for the $AR_0$ and $BR_0$ links are summarised in Table \ref{table:cdf_snr}. It can be seen that for the two network deployments, the probability that $\Gamma_{AR_0}$ is low or medium is 0. This is because the base station is transmitting at very high power and in the worst-case the relay is only 600 metres away, which is within the boundary derived in Section \ref{sec:pncb_algorithm}. On the other hand, since the mobile is transmitting at low power, it is quite likely that $\Gamma_{BR_0}$ is low (10\% for $s$ = 100m and 23\% for $s$ = 600m). As expected, the likelihood of a low $\Gamma_{BR_0}$ is greater for a network with fewer relays.

\subsection{Experimental Setup}
\label{sec:exp_setup}
\begin{figure}[h!]\centering\includegraphics[width=0.7\textwidth]{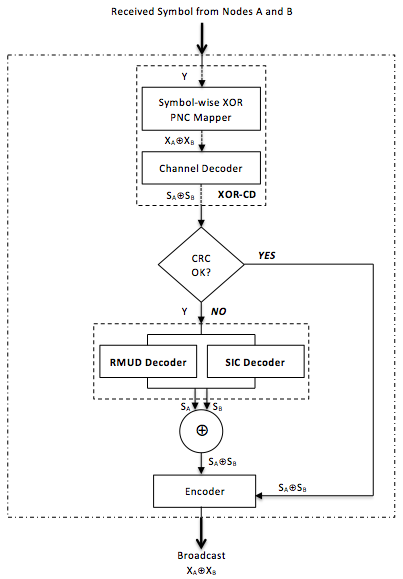}\caption{Decoder Architecture at the Relay}\label{fig:relay_architecture}\end{figure}

We use the implementation of PNC on the USRP platform, detailed in \cite{Lu2013a}, to analyse the decoding performance for various SNRs. The physical layer is based on OFDM and the cyclic prefix (CP) is used to resolve symbol asynchrony and prevent inter-symbol interference. The details of this and the frame format can be found in \cite{Lu2013a}. The modulation and coding schemes used by the end-nodes (A and B) are BPSK and convolutional coding as defined in the LTE standard \cite[Chapter 10]{Sesia2009}. Since the decode-and-forward strategy is considered in this paper, link-by-link channel-coded PNC is implemented where the relay decodes the superimposed channel coded symbols from $A$ and $B$ and re-encodes it before broadcasting. Let the source symbols from nodes $A$ and $B$ be $S_A$ and $S_B$ respectively. After channel coding, the symbols $X_A$ and $X_B$ are transmitted in the multiple-access phase. The relay receives the superimposed coded symbols from $A$ and $B$ corrupted by noise. It then tries to decode $S_A \oplus S_B$ before re-encoding it for the broadcast phase. The XOR-Channel Decoder (XOR-CD) \cite{Lu2013a} is used at the relay. In XOR-CD, an XOR mapping of the received symbol $Y$ to the transmitted network-coded symbol $X_A \oplus X_B$ is first performed, followed by channel decoding to obtain $S_A \oplus S_B$. The channel decoder at the relay implements the Viterbi algorithm \cite{Lu2013a}. Error checking is also performed to verify if the decoding was successful. The IEEE CRC-32 (cyclic redundancy check) function is modified in the implementation so that CRC($S_A \oplus S_B$) = CRC($S_A$) $\oplus$ CRC($S_B$) \cite{Lu2013b}.

If the channel decoding is unsuccessful (i.e. CRC check fails), two additional decoders are used to decode the individual source symbols $S_A$ and $S_B$. The first decoder is based on reduced-constellation multi-user detection (RMUD) and the second is based on successful interference cancellation (SIC). In RMUD, the number of constellation points to consider for decoding is reduced by adopting the log-max approximation. In BPSK, for instance, the four possible constellation points ($\pm 1,\pm 1$) are reduced to two, the details of which can be found in \cite{Lu2013}. In SIC, the stronger signal is first decoded and the estimate of that is subtracted from the received signal to then decode the weaker one. If decoding is successful, both decoders would output $S_A$ and $S_B$ and they are then combined to form $S_A \oplus S_B$. Figure \ref{fig:relay_architecture} gives an overview of the decoder architecture at the relay. 

\subsection{Decoding Performance}
\label{sec:decoding_perf}
Figure \ref{fig:pnc_decoding_accuracy} shows the decoding performance of XOR-CD alone and also the combined XOR-CD, RMUD and SIC decoders. The x-axis shows the received SNRs of the two end nodes at the relay. This is represented as an ordered tuple of the form $\left(\Gamma_{BR_0},\Gamma_{AR_0}\right)$ in dBs. It can be observed from the graph that at low SNRs, the decoding performance of XOR-CD is very low. It can also be observed that the decoding success rate improves when the received SNRs are imbalanced. For instance, for the (7,7.5) dB SNR pair, the success rate of XOR-CD is only about 10\% and the success rate improves to about 33\% for the (7,9) dB SNR pair. Similarly, the decoding performance of (7,9.5) dB is significantly greater than (7.5,7.5) dB. Since all the nodes are transmitting at maximum power, power control to balance the SNRs could be detrimental to the decoding performance, especially at low SNRs. At medium to high SNRs, the decoding performance of XOR-CD is between 88-95\%. The success rate of XOR-CD could be further improved by using more advanced coding techniques.
\begin{figure}[h!]\centering\includegraphics[width=0.85\textwidth]{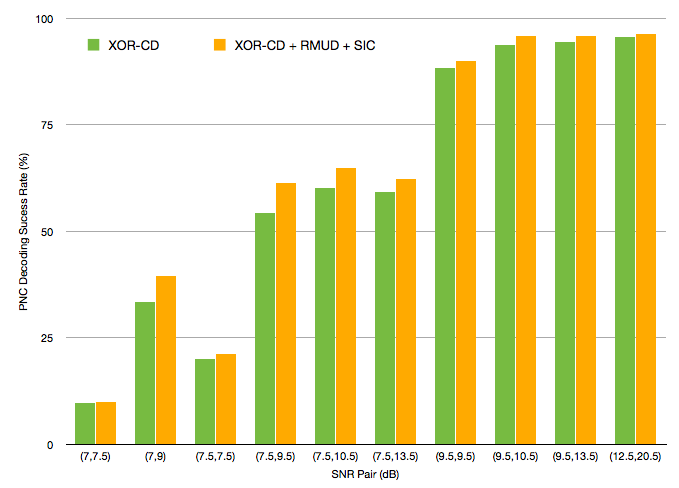}\caption{PNC Decoding Performance at Relay}\label{fig:pnc_decoding_accuracy}\end{figure}

The gain of using RMUD and SIC decoders over XOR-CD can be observed to be greater at low and medium SNRs and not so significant at high SNRs. This gain is quantified in Table \ref{table:rmud_sic}. We will look at the fourth row in the table to illustrate how the data can be analysed. For the SNR pair (7.5,9.5) dB, the gain provided by the RMUD and SIC decoders over XOR-CD is 6.9\%. The contributions of RMUD and SIC toward this 6.9\% gain are 10\% and 97\% respectively. The final column shows that both RMUD and SIC get 7\% in common for successfully decoding the individual source symbols $S_A$ and $S_B$.

We can now observe from Table \ref{table:rmud_sic} that at low and balanced SNRs, the success rate gain from RMUD and SIC is very low (between 0.1 and 1.2\%). When the SNRs are imbalanced, the contribution of RMUD and SIC is greater (between 4.6 and 6.9\%), especially when the SNR from one of the nodes is low. It can also be observed that the contribution of SIC is significantly greater when there is an imbalance in SNRs. This is expected since SIC is designed to differentiate and decode the strong and weak signals. 

\begin{table}[!h]
\caption{Contribution of RMUD and SIC Decoders}
\centering
\begin{tabular}{| c || c || c | c | c |}
\hline
\multirow{2}{*}{\textbf{SNR Pair (dB)}} & \textbf{Success Rate Gain from} & \multicolumn{3}{ c |}{\textbf{Contribution towards Gain}} \\ \cline{3-5}
 & \textbf{RMUD + SIC (\%)} & \textbf{RMUD (\%)} & \textbf{SIC (\%)} & \textbf{RMUD $\cap$ SIC (\%)} \\ \hline \hline
(7,7.5) & 0.1 & 100 & 0 & 0 \\ \hline
(7,9) & 6.1 & 0 & 100 & 0 \\ \hline
(7.5,7.5) & 1.2 & 25 & 75 & 0 \\ \hline
(7.5,9.5) & 6.9 & 10 & 97 & 7 \\ \hline
(7.5,10.5) & 4.6 & 9 & 96 & 5 \\ \hline
(7.5,13.5) & 3.1 & 15 & 87 & 2 \\ \hline
(9.5,9.5) & 1.7 & 30 & 90 & 20 \\ \hline
(9.5,10.5) & 2.2 & 47 & 76 & 23 \\ \hline
(9.5,13.5) & 1.5 & 27 & 91 & 18 \\ \hline
(12.5,20.5) & 0.8 & 27 & 86 & 13 \\ \hline
\end{tabular}
\label{table:rmud_sic}
\end{table} 

These experimental results have practical implications and can be used by the network operator to make a tradeoff between cost and complexity. If cost in terms of relay deployment is an issue, then there will be fewer relays in the network resulting in the likelihood of a low $\Gamma_{BR_0}$ to be higher. Then, the highly complex relay with XOR-CD+RMUD+SIC decoders have to be deployed to improve the decoding performance. If on the other hand complexity is an issue and the relay consists of XOR-CD only, then more relays have to be deployed to reduce the likelihood of a low $\Gamma_{BR_0}$.

\section{Conclusions and Future Work}
\label{sec:conclusions}
This paper applies PNC in a heterogeneous cellular network in a single cell. For bidirectional traffic between the base station and the user, a relay-selection algorithm is proposed where uneven time allocations for the multiple-access phase and the broadcast phase are adopted to maximize the overall achievable data exchange rate. The optimisation problem is shown to be log-concave and relay selection is based on the gradient-ascent algorithm. Compared to the widely applied selection-cooperation technique, the proposed algorithm performs significantly better for all user locations and network deployments. 

The decoding performance of PNC with imbalanced SNRs is then studied on the software radio platform. For the setup considered, the experimental results show that the decoding success rate improves when the received SNRs are imbalanced with channel coding. Since all the nodes are transmitting at maximum power, the results show that any power control to balance the SNRs could be detrimental to the decoding performance, especially at low SNRs. Two additional decoders based on multiuser detection and successive interference cancellation are also studied and when combined with the channel decoder, it is shown to improve the decoding performance at low and medium SNRs.

In the future, this study will be extended to a multi-cell setting with higher-order modulations and advanced coding techniques. The problem of resource allocation in the presence of fading will also be considered. The studies in this paper assumed an accurate estimate of the user location based on the received SNR. An accuracy study of these estimates based on real network data will be undertaken in the future as well.

\section*{Acknowledgement}
The authors would like to acknowledge Toshiba Research Europe Ltd. (TREL) and the U.K. Research Council for supporting the work done in this paper through the Dorothy Hodgkin Postgraduate Award. The authors would also like to thank the Worldwide Universities Network (WUN) Research Mobility Programme (RMP) for enabling the collaboration between the University of Bristol and The Chinese University of Hong Kong. This work is also partially supported by the General Research Funds (Project No. 414812) and AoE grant E-02/08, established under the University Grant Committee of the Hong Kong Special Administrative Region, China. 

\bibliographystyle{IEEEtran}
\bibliography{library}

\end{document}